\documentclass[letterpaper, 10 pt, conference]{ieeeconf}  % Comment this line out if you need a4paper

\IEEEoverridecommandlockouts                              % This command is only needed if 
\usepackage{graphics} % for pdf, bitmapped graphics files
\usepackage{epsfig} % for postscript graphics files
\usepackage{times} % assumes new font selection scheme installed
\usepackage{amsmath} % assumes amsmath package installed
\usepackage{amssymb}  % assumes amsmath package installed
\usepackage{graphicx}
\usepackage{epstopdf}
\usepackage{xcolor}
\newtheorem{definition}{Definition}
\newtheorem{lemma}{Lemma}
\newtheorem{assumption}{Assumption}
\newtheorem{remark}{Remark}
\newtheorem{proposition}{Proposition}
\newtheorem{theorem}{Theorem}

\title{\LARGE \bf
Distributed Safety Critical Control among Uncontrollable Agents Using Reconstructed Control Barrier Functions
}

\author{Yuzhang Peng, Wei Wang, Jiaqi Yan, and Mengze Yu% <-this % stops a space
\thanks{*This work was supported by National Natural Science Foundation of China under Grants 62573023 and 62373019,
and the Fundamental Research Funds for the Central Universities (501RCQD2025103003, 501XSKC2025103001).
}% <-this % stops a space
\thanks{Yuzhang Peng and Mengze Yu are with the School of Automation Science and Electrical Engineering, Beihang University, Beijing 100191, China.
        {\tt\small pengyz@buaa.edu.cn}; {\tt\small yumengze0214@126.com}}%
\thanks{Wei Wang and Jiaqi Yan are with the School of Automation Science and Electrical Engineering, Beihang University, Beijing 100191, China and the Hangzhou Innovation Institute, Beihang University, Hangzhou 310051, China.
        {\tt\small w.wang@buaa.edu.cn}; {\tt\small jqyan@buaa.edu.cn}}%
}

\begin{document}

\maketitle
\thispagestyle{empty}
\pagestyle{empty}

%%%%%%%%%%%%%%%%%%%%%%%%%%%%%%%%%%%%%%%%%%%%%%%%%%%%%%%%%%%%%%%%%%%%%%%%%%%%%%%%
\begin{abstract}

This paper investigates the distributed safety critical control for multi-agent systems (MASs) in the presence of uncontrollable agents with uncertain behaviors. 
To ensure system safety, the control barrier function (CBF) is employed in this paper. 
However, a key challenge is that the CBF constraints are coupled when MASs perform collaborative tasks, which depend on information from multiple agents and impede the design of a fully distributed safe control scheme. To overcome this, a novel reconstructed CBF approach is proposed. In this method, the coupled CBF is reconstructed by leveraging state estimates of other agents obtained from a distributed adaptive observer. 
Furthermore, a prescribed performance adaptive parameter is designed to modify this reconstruction, ensuring that satisfying the reconstructed CBF constraint is sufficient to meet the original coupled one. Based on the reconstructed CBF, we design a safety-critical quadratic programming (QP) controller and prove that the proposed distributed control scheme rigorously guarantees the safety of the MAS, even in the uncertain dynamic environments involving uncontrollable agents. The effectiveness of the proposed method is illustrated through a simulation.

\end{abstract}

%%%%%%%%%%%%%%%%%%%%%%%%%%%%%%%%%%%%%%%%%%%%%%%%%%%%%%%%%%%%%%%%%%%%%%%%%%%%%%%%
\section{INTRODUCTION}
In recent years, safety critical control has garnered significant attention for its applications in fields such as robotics \cite{Cavorsi2024}, autonomous driving \cite{Shen2025}, and aerospace \cite{Yang2025}. Beyond established methods such as artificial potential fields (APFs) \cite{Khatib1986}, prescribed performance control (PPC) \cite{Bechlioulis2008}, and nonovershooting control \cite{Kristic2006}, the use of control barrier functions (CBFs) to ensure system safety represents an emerging research direction. The standard CBF-based approach adopts a quadratic program (QP) as the controller to compute the control inputs that strictly satisfies the CBF-induced safety constraint, thereby guaranteeing system safety \cite{Ames2017}.

Due to the widespread application of multi-agent systems (MAS) in collaborative tasks, ensuring their safety using CBFs has also become an active area of research. The works in \cite{Charitidou2023} and \cite{Wang2017} investigate the cooperative control problem for MASs under safety or task constraints using a CBF-based approach. However, both of these studies adopt a centralized control architecture. A distributed CBF-QP control framework is proposed in \cite{Lindemann2020} to address the control problem under collaborative task constraints. Nevertheless, this approach assumes that the MAS for achieving the collaborative task is fully connected. Therefore, designing distributed safety controllers based on CBFs to address the cooperative control problem still remains a challenge.

The primary challenge in designing distributed CBF-based controllers is that the CBF often involves the states of multiple agents, resulting in a \textit{coupled CBF}. When a standard QP controller is adopted, this coupled CBF induces a \textit{coupled constraint} that involves the states and control inputs of multiple agents. Consequently, it becomes infeasible for an individual agent to compute a control input that satisfies this constraint using only its local information. 

One approach to overcome the issue of coupled CBFs and enable distributed control is the constraint decomposition approach proposed in \cite{Tan2025,Mestres2024}. The core assumption therein is that a coupled constraint can be decomposed into multiple local constraints, each relying solely on information available to the respective agent. Consequently, if the control input of every agent satisfies its respective local constraint, it can be guaranteed that the original coupled constraint is also satisfied. While the constraint decomposition approach allows for the design of fully distributed QP controllers, a drawback is that not all coupled CBF constraints can be transformed into local ones through a simple additive split as used in \cite{Tan2025,Mestres2024}.

Another challenge arises from the presence of \textit{uncontrollable agents} with uncertain behavior in the operating environment of MASs \cite{Brafman1996,Yu2023}. Prominent examples include pedestrians in multi-robot navigation scenarios and human-driven vehicles in autonomous driving contexts. When uncontrollable agents are present in the environment, the constraint decomposition approach proposed in \cite{Tan2025,Mestres2024} becomes inapplicable. This is because the method requires all agents involved in a coupled constraint to satisfy their respective decomposed local constraints. However, we can only design control inputs for \textit{controllable agents}, whereas the inputs of uncontrollable agents are entirely unknown and cannot be prescribed. This makes it impossible to satisfy the local constraints corresponding to the uncontrollable agents, which in turn renders the design of distributed safety controllers more challenging.

\textbf{Contributions:} Motivated by the aforementioned limitations, we propose a novel reconstructed CBF approach to resolve the challenge of coupled CBF constraints in distributed safety critical control. 
\begin{itemize}
  \item By leveraging state estimates from distributed adaptive observers, this method transforms coupled, global CBFs into local ones. We further modify the reconstructed CBFs with a prescribed performance adaptive parameter, which guarantees that satisfying the reconstructed CBF constraints is sufficient to satisfy the original global one.
  \item Prior works, such as \cite{Tan2025,Mestres2024}, require the presupposition that all agents involved in a coupled constraint are controllable. In contrast, the reconstructed CBFs method enables controllable agents to compensate for the uncertain behaviors of uncontrollable ones. This relaxes the requirement for global constraint satisfaction, such that instead of requiring all agents to enforce the coupled CBF constraint, only a subset of controllable agents needs to satisfy their reconstructed local constraints.
\end{itemize}

The remainder of this paper is organized as follows: Section \ref{sec:2} provides the preliminaries and describes the problem. In Section \ref{sec:3}, a distributed safety critical control framework based on reconstructed CBFs is proposed. Section \ref{sec:4} demonstrates the effectiveness of the proposed method through a simulation example. A conclusion is drawn in Section \ref{sec:5}.

\textbf{Notation}: 
Symbol $\mathbb{R}$ represents the real number set. $\mathbb{R}^{n}$ denotes the $n$-dimensional real vector space and $\mathbb{R}^{m \times n}$ denotes the $m \times n$-dimensional real matrix space. 
$I$ represents an identity matrix with an appropriate dimension. Inequalities $X \succ 0$ and $Y \succeq 0$ indicate that matrices $X$ and $Y$ are positive definite and positive semi-definite, respectively. $\mathbf{1}$ and $\mathbf{0}$ refer to vectors of proper dimensions with all entries to be $1$ or $0$, respectively. The Kronecker product of matrices $M\in \mathbb{R}^{m\times n}$ and $N\in \mathbb{R}^{p\times q}$ is denoted by $M\otimes N \in \mathbb{R}^{mp\times nq}$.

\section{PRELIMINARIES AND PROBLEM STATEMENT}
\label{sec:2}

\subsection{Graph Theory}
In this paper, different controllable agents are connected through an undirected graph $\mathcal{G} = \left\{\mathcal{N},\mathcal{E}\right\} $, which consists of agents $\mathcal{N} = \{1,2,\dots,N\}$ and edges $\mathcal{E} = \{(i,j)|i,j\in\mathcal{N}\}$. Since $\mathcal{G}$ is undirected, if the edge $(i,j) \in \mathcal{E}$, then $(j,i) \in \mathcal{E}$. Moreover, $(i,j) \in \mathcal{E}$ indicates that agent $j$ can obtain the state information of agent $i$. In this case, agent $j$ is referred to as a neighbor of agent $i$, and vice versa. Self edge is disallowed in this study, i.e., $(i,i) \notin  \mathcal{E}, \forall i \in \mathcal{N}$. The adjacency matrix $\mathcal{A} = [a_{ij} ] \in \mathbb{R}^{N\times N}$ is given by $a_{ij} = 1$ if $(i, j) \in \mathcal{E}$, and $a_{ij} = 0$, otherwise. The degree matrix is denoted by $\mathcal{D} = \text{diag} \left\{d_1,\dots,d_N \right\} $, where $d_i = \sum_{j=1}^{N} a_{ij}$. The Laplacian matrix $\mathcal{L}$ is provided by $\mathcal{L} = \mathcal{D} - \mathcal{A}$. The undirected graph $\mathcal{G}$ is connected if there is a sequence of connected edges between any agents $i$ and $j$. The set of uncontrollable agents is defined as $\mathcal{V} = \{N+1,\dots,N+V\}$. Similarly, if the controllable agent $i\in \mathcal{N}$ is directly connected with the uncontrollable agent $l \in \mathcal{V}$, which means that agent $i$ is aware of the state information of agent $l$, we have $a_{il} = 1$ and $a_{il} = 0$, otherwise. 
Let $\mathcal{N}^+ := \mathcal{N} \cup \mathcal{V}$ be the set that includes all controllable agents and uncontrollable agents.
For any $i \in\mathcal{N}$ and $j \in\mathcal{N}^+ $, define $\mathcal{H}_j = \mathcal{L} + \mathcal{B}_j$ with $\mathcal{B}_j = \text{diag}\left\{b_{1j},\dots,b_{Nj}\right\} $, where $b_{ij} = a_{ij}$ if $i\neq j$ and $b_{ii} = 1$ otherwise.

\subsection{System Dynamics and Control Barrier Functions}
Consider a class of nonlinear MASs as
\begin{equation}
  \label{eq:sys}
  \dot x_i = f(x_i) + g(x_i)u_i, \, i\in \mathcal{N}^+,
\end{equation}
where $x_i = \left[x_{i,1},\dots,x_{i,n}\right]^T \in \mathbb{R}^n$ and $u_i \in \mathbb{R}^m$ are the state and control input of agent $i$, respectively. The function $f:\mathbb{R}^n\rightarrow \mathbb{R}^n$ is globally Lipschitz continuous with an unknown Lipschitz constant $L_1$, while the function $g:\mathbb{R}^n\rightarrow \mathbb{R}^{n \times m}$ is locally Lipschitz continuous. Let $x = [x_1^T,\dots,x_{N+V}^T]^T$ and $u = [u_1^T,\dots,u_{N+V}^T]^T$ denote the stacked state vector and stacked control input vector of all agents in $\mathcal{N}^+$, respectively. Note that the control input $u_l$ of the uncontrollable agent $l \in \mathcal{V} \subseteq \mathcal{N}^+$ is unknown and cannot be designed.

In the CBF-based standard safety critical control method, the fundamental objective is to constrain the system state within a predefined \textit{safe set}. Given the MAS (\ref{eq:sys}) and a continuously differentiable function $h(x):\mathbb{R}^{(N+V)n} \rightarrow \mathbb{R}$, the safe set $\mathcal{C}$ for (\ref{eq:sys}) is defined as the zero superlevel set of function $h$ as follows:
\begin{align}
  \label{eq:safe_set}
  \mathcal{C}:=\{x\in\mathbb{R}^{(N+V)n} | h(x)\geq0\}.
\end{align}
The interior and boundary of $\mathcal{C}$, respectively denoted by $\mathrm{Int}(\mathcal{C})$ and $\partial\mathcal{C}$, are defined as follows:
\begin{align}
  & \mathrm{Int}(\mathcal{C}):=\{x\in\mathbb{R}^{(N+V)n} | h(x)>0\}, \notag \\
  & \partial\mathcal{C}:=\{x\in\mathbb{R}^{(N+V)n} | h(x)=0\}.
\end{align}
Afterwards, the definition of CBFs is introduced as follows:
\begin{definition}[Zeroing CBFs, ZCBFs \cite{Ames2017}] \! The continuously differentiable function $h(x)$ is a ZCBF if there exists an extended class $\mathcal{K}$ function $\alpha(\cdot)$ for all $x(t) \in \mathcal{C}$, such that
  \begin{align}
    \label{eq:uzcbf}
    \sup_{u\in\mathbb{R}^{(N+V)m}} & \left[\sum_{i=1}^{N+V}\frac{\partial h}{\partial x_i}\left(f(x_i) + g(x_i)u_i\right)\right] \geq -\alpha(h(x)).
  \end{align}
\end{definition}
Then, the following lemma is established for ZCBFs.
\begin{lemma}[\cite{Ames2017} ]
  \label{lem:1} 
  For a given ZCBF $h(x)$ and MAS (\ref{eq:sys}), if the control input satisfies $u\in\mathcal{U}_Z$, where $\mathcal{U}_Z$ is the set of control input satisfying (\ref{eq:uzcbf}), the safe set $\mathcal{C}$ is forward invariant.
\end{lemma}

\subsection{Problem Statement}
In research related to CBFs, QP is commonly employed to synthesize a controller that ensures the control input strictly satisfies safety constraints. For system (\ref{eq:sys}), to guarantee that constraint (\ref{eq:uzcbf}) is strictly satisfied, the standard QP-based controller is formulated as follows:
\begin{subequations}
  \label{eq:standardQP}
\begin{align}
  & u = \mathop{\arg \min} \limits_{ {\nu}\in \mathbb{R}^{(N+V)m}} \,
    \frac{1}{2} (\nu - u_{\text{nom}}) ^T W (\nu - u_{\text{nom}})   \\ 
  \label{eq:coupled_constraint}
  & \text{s.t.}  \quad \sum_{i=1}^{N+V}\frac{\partial h}{\partial x_i}\left(f(x_i) + g(x_i)u_i\right) \geq -\alpha(h(x)).
\end{align}
\end{subequations}
The controller (\ref{eq:standardQP}), also referred to as the \textit{safety filter}, is used to modify the nominal control input $u_{\text{nom}}\in \mathbb{R}^{(N+V)m}$, which is designed to achieve a primary control task (e.g. trajectory tracking) without considering safety constraints. The remaining terms, $\nu$ and $W\in \mathbb{R}^{(N+V)m \times (N+V)m}$, are the optimization variable and a positive definite weighted matrix, respectively. According to Lemma \ref{lem:1}, (\ref{eq:standardQP}) is able to strictly ensure the safety of the system while achieving the control task corresponding to $u_{\text{nom}}$ as effectively as possible. However, the presence of the coupled constraint (\ref{eq:uzcbf}) renders the controller (\ref{eq:standardQP}) centralized. Besides, it is not possible to solve for the control input that satisfies (\ref{eq:uzcbf}) using only the information from a single agent. Moreover, due to the presence of uncontrollable agents, $u_i$ for $i\in\mathcal{V}$ cannot be designed, which renders (\ref{eq:standardQP}) inapplicable and exacerbates the difficulty in designing a distributed safety controller.

To overcome the aforementioned challenges, the control objective of this paper is as follows:
\begin{itemize}
  \item Design a distributed QP controller for each controllable agent $i \in \mathcal{N}$, which rigorously guarantees the safety constraints while utilizing only the local information of the agent;
  \item The designed distributed safety controller remains applicable even when the safety constraints involve uncontrollable agents, i.e., $\frac{\partial h}{\partial x_i} \not\equiv \boldsymbol{0}^T$ for $i\in\mathcal{V}$.
\end{itemize}

\begin{remark}
  \label{re:1}
  Although our analysis, as presented in (\ref{eq:coupled_constraint}), focuses on the case where all controllable agents are subject to a single, common coupled CBF constraint, the proposed method can be readily extended to scenarios involving multiple heterogeneous coupled CBF constraints among agents. We demonstrate this capability with a case study in Section \ref{sec:4}.
\end{remark}

To achieve the control objectives, the following assumptions are imposed.

\begin{assumption} 
  \label{ass:1} 
  % \cite{Deng2020}
The undirected graph $\mathcal{G}$ is connected, and the state of uncontrollable agent $l$ for any $l\in\mathcal{V}$ can be obtained by at least one controllable agent.
\end{assumption}
\begin{assumption}
  \label{ass:2} 
  % \cite{Lindemann2019a,Lindemann2019b,Ze2025,Chen2024}
  The control effectiveness $g(x_i)u_i$ is bounded, i.e., $\left\lVert g(x_i)u_i\right\rVert \leq D$ for any $i \in \mathcal{N}^+$, where $D>0$ is an unknown upper bound.
\end{assumption}
\begin{assumption}
  \label{ass:3} 
  % \cite{Lindemann2019a,Lindemann2019b,Ze2025,Chen2024}
  The derivative of the function $h(z)$, i.e., $\frac{d h(z)}{d z}$, is designed to be locally Lipschitz continuous.
\end{assumption}
\begin{assumption}
  \label{ass:4}
  For the CBF $h(x)$, both its function form and its value are known to any controllable agent $i\in\mathcal{N}$.
\end{assumption}
\begin{assumption}
  \label{ass:5} %\cite{Shen2025}
  At the initial time $t=0$, each agent $i$ is in the interior of the safe set $\mathcal{C}$, i.e., $h(x(0)) >0$.
\end{assumption}
\begin{remark}
  Assumption \ref{ass:1} relaxes the fully connected topology required in \cite{Lindemann2020} while maintaining the necessary connectivity for cooperative tasks like consensus \cite{Deng2020}. Assumption \ref{ass:2} standardly bounds control effort \cite{Lindemann2019b,Chen2024}, and Assumption \ref{ass:3} is easily satisfied.
  This work departs from \cite{Tan2025} which assumes that coupled constraints can be decomposed into local constraints known to individual agents. Instead, Assumption \ref{ass:4} posits that an agent must be aware of both the safety objective (the functional form of $h$) and its current safety status (the value of $h$). This is reasonable. For instance, in collision avoidance ($h(x_1,x_2):=\lVert x_1 - x_2\rVert^2-d^2$, where $d>0$ is a safe distance), agents do not need each other's full state; the CBF value (relative distances) obtained via local sensing (e.g., ultra-wideband) are sufficient. Finally, Assumption \ref{ass:5} necessitates initial safety to prevent immediate constraint violations caused by uncertainties \cite{Shen2025}.
\end{remark}

The following lemmas will be utilized in the subsequent sections.
\begin{lemma}[\cite{Godsil2013}]
  \label{lem:2} 
    Under Assumption \ref{ass:1}, it is clear that the non-zero diagonal matrix $\mathcal{B}_j \succeq 0$ for all $j\in\mathcal{N}^+$. Thus, it can be obtained that $\mathcal{H}_j \succ 0$, with the eigenvalues $\lambda_{j,l}$ of $\mathcal{H}_j$ satisfying $0<\lambda_{j,1}\leq\dots\leq\lambda_{j,N}$. 
\end{lemma}

\begin{lemma}[\cite{Li2022}]
  \label{lem:3} 
    The following inequality holds for any $\epsilon>0$ and any $\varpi \in \mathbb{R}$:
    \begin{equation}
      |\varpi| \leq \frac{\varpi^2}{\sqrt{\varpi^2 + \epsilon^2}} + \epsilon.
    \end{equation}
\end{lemma}

\section{MAIN RESULTS} 
\label{sec:3}
The distributed control scheme proposed in this paper will be presented in this section.

\subsection{Design of Distributed Adaptive Observer}
Noting that constraint (\ref{eq:coupled_constraint}) incorporates the states and control inputs of multiple agents, the implementation of distributed control becomes challenging due to the incomplete information acquisition by a single agent. To overcome this issue, 
we first design a distributed observer for each agent $i\in \mathcal{N}$, which allows it to estimate the states of the other agents.
\begin{align}
  \label{eq:observer}
  & \dot {\hat x}_{i,l} = f(\hat x_{i,l}) - \hat \delta_{i,l}\xi_{i,l}, \notag  \\
  & \xi_{i,l} = \sum_{j=1}^{N}a_{ij}\left(\hat x_{i,l} - \hat x_{j,l}\right) + b_{il}\left(\hat x_{i,l} - x_l\right), \notag \\
  & \dot {\hat \delta}_{i,l} = 2\xi_{i,l}^T P_l \xi_{i,l} - \sigma_{l} {\hat \delta}_{i,l}, \, \, l=1,\dots,N+V,
\end{align}
where $\hat x_{i,l}$ is the estimate of the state $x_l$ by agent $i$, $P_l \in \mathbb{R}^{n\times n}$ is a symmetric positive definite matrix, $\sigma_l$ is a positive constant. $\hat \delta_{i,l} $ is an adaptive parameter used to estimate the constant $\delta_l$, which satisfies 
\begin{equation}
  \label{eq:obs_LMI}
  2\lambda_{l,N}P_l^2 + \frac{L_1^2}{\lambda_{l,1}}I - 2\lambda_{l,1}\delta_lP_l + \sigma_l P_l \prec 0.
\end{equation}
Afterwards, the following proposition can be established. 
\begin{proposition} \label{pro:1}
  Consider the MAS described by (\ref{eq:sys}) with the distributed adaptive observers (\ref{eq:observer}). Provided that Assumptions \ref{ass:1} and \ref{ass:2} are satisfied, it can be concluded that 

  1) all estimation errors $\tilde x_{i,l} := \hat x_{i,l} - x_{l}$, $\dot {\tilde x}_{i,l}$ and adaptive parameters $\hat \delta_{i,l}$ remain globally uniformly bounded with $ \limsup _{t\rightarrow\infty} ||\tilde x_{i,l}(t)|| \leq \sqrt{\frac{\Xi_l}{\lambda_{l,1}\lambda_{\min}(P_l)\sigma_l}}$, $i \in \mathcal{N}$, $l \in \mathcal{N}^+$, where $\Xi_l:=\frac{\sigma_l}{2} \sum_{i=1}^{N} \delta_{l}^2 + ND^2$;

  2) if the signal $x_l$ for $l\in\mathcal{N}^+$ is bounded, then all signals of the observer (\ref{eq:observer}) are bounded.
\end{proposition}
\begin{proof}
  The time derivative of $\tilde x_{i,l}$ is calculated as
  \begin{equation}
    \label{eq:dtilde x_{i,l}}
    \dot {\tilde{x}}_{i,l} = \tilde f_{i,l} - d_l - \hat \delta_{i,l}\xi_{i,l},
  \end{equation}
  where $\tilde f_{i,l} := f(\hat x_{i,l}) - f(x_{l})$ and $d_l := g({x}_{l})u_l$. Consequently, define $\tilde{x}_l := \left[\tilde x_{1,l}^T,\dots,\tilde x_{N,l}^T\right]^T $. Consider the following Lyapunov function 
  \begin{equation}
    \label{eq:V_l}
    V_l = \tilde{x}_l^T\left(\mathcal{H}_l\otimes P_l\right)\tilde{x}_l + \frac{1}{2}\sum_{i=1}^{N}\tilde \delta_{i,l}^2 ,
  \end{equation}
  where $\tilde \delta_{i,l} = \delta_l - \hat \delta_{i,l}$. It can be derived that 
  \begin{equation}
    \label{eq:dtilde x_l}
    \dot{\tilde x}_l = \tilde F_l - \boldsymbol{1}\otimes d_l - \hat{\delta}_l\left(\mathcal{H}_l \otimes I\right) \tilde{x}_l,
  \end{equation}
  where $\tilde{F}_l = \left[\tilde f_{1,l}^T, \dots,\tilde f_{N,l}^T\right]^T $ and $\hat \delta_l = \mathrm{diag}\{ \underbrace{\hat{\delta}_{1,l}, \dots ,\hat{\delta}_{1,l},} \limits_{n}\underbrace{\hat{\delta}_{2,l},\dots ,\hat{\delta}_{2,l},}\limits_{n} \dots \underbrace{,\hat{\delta}_{N,l}, \dots ,\hat{\delta}_{N,l}}\limits_{n} \}_{Nn \times Nn}$. Thus, the derivative of $V_l$ along (\ref{eq:dtilde x_l}) is obtained as
  \begin{align}
    \label{eq:dV_l}
    \dot V_l = & 2\tilde{x}_l^T\left(\mathcal{H}_l\otimes P_l\right)\tilde{F}_l - 2\tilde{x}_l^T\left(\mathcal{H}_l\otimes P_l\right)\left(\boldsymbol{1}\otimes d_l\right) \notag \\
    & - 2\tilde{x}_l^T\left(\mathcal{H}_l\otimes P_l\right)\hat{\delta}_l\left(\mathcal{H}_l \otimes I\right) \tilde{x}_l - \sum_{i=1}^{N}\tilde \delta_{i,l}\dot {\hat \delta}_{i,l} \notag \\
    \leq & 2\tilde{x}_l^T\left(\mathcal{H}_l^2\otimes P_l^2\right)\tilde{x}_l + ||\tilde{F}_l||^2 +||\boldsymbol{1}\otimes d_l||^2 \notag \\
    & - 2\sum_{i=1}^{N}\hat \delta_{i,l}\xi_{i,l}^TP_l\xi_{i,l} - \sum_{i=1}^{N}\tilde \delta_{i,l} \dot {\hat \delta}_{i,l} .
  \end{align}
  Based on the global Lipschitz continuity of the function $f(\cdot)$ and Assumption \ref{ass:4}, it can be deduced that
  \begin{align}
    \label{eq:tilde F_l}
    ||\tilde{F}_l||^2 = \sum_{i=1}^{N}||\tilde{f}_{i,l}||^2\leq L_1^2||\tilde{x}_{l}||^2, \, ||\boldsymbol{1}\otimes d_l||^2\leq ND^2.
  \end{align}
  Substituting (\ref{eq:tilde F_l}) into (\ref{eq:dV_l}) and using Lemma \ref{lem:2} yields:
  \begin{align}
    \label{eq:dV_l2}
    \dot V_l 
    \leq &2\tilde{x}_l^T\left(\mathcal{H}_l^2\otimes P_l^2\right)\tilde{x}_l + L_1^2\tilde{x}_{l}^T\tilde{x}_{l} + ND^2 \notag \\
    & - 2\sum_{i=1}^{N}\left(\delta_l-\tilde \delta_{i,l}\right) \xi_{i,l}^TP_l\xi_{i,l} - \sum_{i=1}^{N}\tilde \delta_{i,l}\dot {\hat \delta}_{i,l} \notag \\
    % \leq & 2\lambda_{l,N}\tilde{x}_l^T\left(\mathcal{H}_l\otimes P_l^2\right)\tilde{x}_l + \frac{L_1^2}{\lambda_{l,1}}\tilde{x}_l^T\left(\mathcal{H}_l\otimes I\right)\tilde{x}_l \notag \\
    % & -2\delta_l\tilde{x}_l^T\left(\mathcal{H}_l^2\otimes P_l\right)\tilde{x}_l + \sigma_l\sum_{i=1}^{N} \tilde \delta_{i,l} \hat \delta_{i,l} + ND^2 \notag \\
    \leq & 2\lambda_{l,N}\tilde{x}_l^T\left(\mathcal{H}_l\otimes P_l^2\right)\tilde{x}_l + \frac{L_1^2}{\lambda_{l,1}}\tilde{x}_l^T\left(\mathcal{H}_l\otimes I\right)\tilde{x}_l \notag \\
    & -2 \lambda_{l,1}\delta_l\tilde{x}_l^T\left(\mathcal{H}_l\otimes P_l\right)\tilde{x}_l + \sigma_l\sum_{i=1}^{N} \tilde \delta_{i,l} \hat \delta_{i,l} + ND^2 .
  \end{align}
  Combining (\ref{eq:obs_LMI}) and  (\ref{eq:dV_l2}), it is deduced that 
  \begin{align}
    \label{eq:dV_l3}
    \dot V_l \leq & - \sigma_l \tilde{x}_l^T\left(\mathcal{H}_l\otimes P_l\right) \tilde{x}_l - \frac{\sigma_l}{2} \sum_{i=1}^{N} \tilde \delta_{i,l}^2 + \frac{\sigma_l}{2} \sum_{i=1}^{N}  \delta_{l}^2 + ND^2 \notag \\
    & = -\sigma_l V_l + \Xi_l .
  \end{align}
  Performing integration on both sides of (\ref{eq:dV_l3}) yields that 
  \begin{align}
    \label{eq:dV_l4}
    V_l(t) \leq \left(V_l(0) - \frac{\Xi_l}{\sigma_l}\right)\mathrm{e}^{-\sigma_l t} + \frac{\Xi_l}{\sigma_l}.
  \end{align}
  From (\ref{eq:V_l}) and (\ref{eq:dV_l4}), it can be concluded that $V_l$, $\tilde x_{i,l}$ and $\tilde \delta_{i,l}$ for $i\in \mathcal{N}, l\in\mathcal{N}^+$ remain bounded. Moreover, $\hat \delta_{i,l}$ and $\xi_{i,l}$ are also bounded. From \eqref{eq:dtilde x_l} and \eqref{eq:tilde F_l}, all $\dot {\tilde x}_{i,l}$ remain bounded. According to (\ref{eq:dV_l4}), it also can be obtain that 
  \begin{equation}
    \limsup _{t\rightarrow\infty} ||\tilde x_{i,l}(t)|| \leq \sqrt{\frac{\Xi_l}{\lambda_{l,1}\lambda_{\min}(P_l)\sigma_l}}.
  \end{equation}
  
  Furthermore, if $x_l$ is bounded, then, based on the previous analysis, it is evident that all the signals of the observer (\ref{eq:observer}) are bounded. In fact, under Assumption \ref{ass:2}, $x_l$ for any $l\in\mathcal{N}^+$ is guaranteed to remain bounded over any finite time horizon. Therefore, without loss of generality, it is assumed that system states are confined within a sufficiently large compact set, thereby ensuring the boundedness of closed-loop signals.
\end{proof}

\subsection{Reconstructed Control Barrier Functions}
Based on the state estimate from the observer (\ref{eq:observer}), the coupled CBF $h(x)$ is then reconstructed for each agent $i\in\mathcal{N}$ as follows:
\begin{align}
  \label{eq:rcbf}
  & \hat h_i(\hat \upsilon_i, \vartheta _i ) = h(\hat \upsilon_i) - \vartheta _i, \notag \\
  & \dot \vartheta _i = \begin{aligned}[t] % 使用[t]选项，使得等号与aligned环境的第一行对齐
      & - c_i \frac{e_i (\rho_i - e_i)}{\rho_i}\varepsilon_i - \frac{\varrho e_i}{\rho_i}(\rho_{0}^i - \rho_{\infty}^i)\mathrm{e}^{-\varrho t} \\
      & - \frac{\rho_i \varepsilon_i}{4e_i(\rho_i - e_i)} - \frac{\hat r_i^2 \left(\left\lVert \frac{d h(\hat \upsilon_i)}{d \hat \upsilon_i}\right\rVert + \left\lVert \dot {\hat {\bar x}}_i \right\rVert \right) \chi_i}{\sqrt{\chi_i ^2\hat r_i^2 + \epsilon_i^2}} , 
    \end{aligned} \notag \\
  & \dot {\hat r}_i = \frac{\gamma_i |\varepsilon_i| \rho_i}{2e_i(\rho_i-e_i)}\left(\left\lVert \frac{d h(\hat \upsilon_i)}{d \hat \upsilon_i}\right\rVert + \left\lVert \dot {\hat {\bar x}}_i \right\rVert \right) - \varsigma_i \hat r_i, 
\end{align}
where $\hat \upsilon_i$ is the vector obtained by replacing the element $x_j$ ($j \neq i, j \in \mathcal{N}^+$) in the vector $x$ with $\hat x_{i,j}$, $\hat {\bar x}_i:= \left[\hat x_{i,1}^T,\dots,\hat x_{i,N+V}^T\right]^T $ and $\chi_i: = \frac{\varepsilon_i\rho_i}{2e_i (\rho_i-e_i)}\left(\left\lVert \frac{d h(\hat \upsilon_i)}{d \hat \upsilon_i}\right\rVert + \left\lVert \dot {\hat {\bar x}}_i \right\rVert \right) $.
The parameters $c_i$, $\varrho$, $\epsilon_i$, $\varsigma_i$ and $\gamma_i$ are positive constants.
Moreover, $\hat h_i$ is the reconstructed CBF, $\vartheta_i$ and $\hat r_i$ are adaptive parameters, $e_i$ is the reconstruction error, given by 
\begin{equation}
  \label{eq:e}
   e_i:= h(x) - \hat h_i(\hat \upsilon_i, \vartheta _i ),
\end{equation}
and $\varepsilon_i$ and $\rho_i$ are respectively defined as 
\begin{align*}
  \varepsilon_i = \frac{1}{2} \ln\left(\frac{e_i}{\rho_i - e_i}\right), \, \rho_i(t) = \left(\rho_{0}^i - \rho_{\infty}^i \right) \mathrm{e}^{-\varrho t} + \rho_{\infty}^i.
\end{align*}
It should be noted that $\vartheta_i$ is an adaptive parameter designed based on PPC theory. Its role is to modify the function $h(\hat \upsilon_i)$ in (\ref{eq:rcbf}) to ensure that when the reconstructed CBF constraint $\hat h_i \geq 0$ is satisfied, the original coupled CBF constraint $h\geq 0$ also holds. This property will be formally proven in the subsequent theorem.
The initial value of $\vartheta_i$, $\rho_{0}^i$ and $\rho_{\infty}^i$ are chosen such that $0<\rho_{\infty}^i<\rho_{0}^i$, $0<e_i(0)<\rho_{0}^i$ and $\hat h_i (\hat \upsilon_i(0), \vartheta _i(0) ) \geq 0$.

Then, regarding the reconstructed CBF (\ref{eq:rcbf}), the following result can be established.
\begin{theorem}
  \label{the:1}
  Considering reconstructed CBF (\ref{eq:rcbf}), under Assumption \ref{ass:1}-\ref{ass:5}, it can be concluded that:
  
  1) The reconstruction error $e_i$ (\ref{eq:e}) satisfies the prescribed performance constraint, i.e., $0<e_i(t)<\rho_i(t)$;

  2) If $\hat h_i \geq 0$ holds for at least one $i \in \mathcal{N}$, then it follows that $h(x)>0$.
\end{theorem}
\begin{proof}
1) Using proof by contradiction, we assume that $e_i$ first exceeds the performance bounds $(0,\rho_i)$ at some finite time $\bar t$, which indicates that $\lim_{t \to \bar t}\varepsilon _i(t) = \infty$. Hence, for $t \in [0,\bar t)$, we have $e_i \in (0,\rho_i)$. The following Lyapunov function is selected as 
\begin{equation}
  \label{eq:V_rcbf}
  V_i = \frac{1}{2}\varepsilon_i^2 + \frac{1}{2 \gamma_i } \tilde{r}_i^2,
\end{equation}
where $\tilde{r}_i = r_i - \hat r_i$. Positive constant $r_i$ satisfies a certain condition to be introduced later. 
Taking the derivative of $V_{i}$ over the interval $[0,\bar t)$ yields
\begin{align}
  \label{eq:dV_1}
  \dot V_i 
  % = & \varepsilon _i\left(\frac{\partial \varepsilon _i}{\partial \rho_i} \dot \rho_i + \frac{\partial \varepsilon _i}{\partial e_i} \dot e_i \right) - \frac{1}{\gamma_i}\tilde r_i \dot{\hat r}_i\notag \\
  = & \varepsilon _i\left[\frac{-1}{2 \left(\rho_i - e_i\right) } \dot \rho_i + \frac{\rho_i}{2 e_i\left(\rho_i - e_i\right) } \right. \notag \\
  &  \left. \times \left(\dot \vartheta_i - \frac{d h(\hat \upsilon _i)}{d \hat \upsilon _i} \dot {\hat \upsilon }_i + \frac{d h(x)}{d x} \dot x \right) \right] - \frac{1}{\gamma_i}\tilde r_i \dot{\hat r}_i \notag \\
  = & \varepsilon _i\left(\frac{\rho_i}{2 e_i\left(\rho_i - e_i\right) }\dot \vartheta_i - \frac{1}{2 \left(\rho_i - e_i\right) } \dot \rho_i\right) \notag \\
  & + \frac{\varepsilon _i\rho_i}{2 e_i\left(\rho_i - e_i\right)}\left[\frac{d h(x)}{d x} \left(\dot x - \dot {\hat \upsilon }_i\right) \right . \notag \\
  & \left. + \left(\frac{d h(x)}{d x} - \frac{d h(\hat \upsilon _i)}{d \hat \upsilon _i}\right)  \dot {\hat \upsilon }_i\right] - \frac{1}{\gamma_i}\tilde r_i \dot{\hat r}_i.
\end{align}
Combining Proposition \ref{pro:1} and Assumption \ref{ass:3}, it is deduced that 
\begin{align}
  \label{eq:bd1}
  & ||\dot x - \dot {\hat \upsilon }_i|| \leq \Delta_{i,1}, \\
  \label{eq:bd2}
  & \left\lVert \frac{d h(x)}{d x} - \frac{d h(\hat \upsilon _i)}{d \hat \upsilon _i}\right\rVert \leq L_2 ||x - {\hat \upsilon }_i|| \leq \Delta_{i,2}, 
\end{align}
where $L_2$ is a Lipschitz constant of $\frac{d h(z)}{d z} $, $\Delta_{i,1}$ and $\Delta_{i,2}$ are positive constants. Defining $r_i:= \max\{\Delta_{i,1},\Delta_{i,2}\}$ and applying Lemma \ref{lem:3}, it can be obtained that 
\begin{align}
  \label{eq:dV_2}
  \dot V_i 
  \leq & \varepsilon _i\left(\frac{\rho_i}{2 e_i\left(\rho_i - e_i\right) }\dot \vartheta_i - \frac{1}{2 \left(\rho_i - e_i\right) } \dot \rho_i\right) \notag \\
  & + \frac{|\varepsilon _i | \rho_i}{2 e_i\left(\rho_i - e_i\right)}\left(\left\lVert \frac{d h(x)}{d x}\right\rVert + \left\lVert \dot {\hat \upsilon }_i\right\rVert \right) r_i - \frac{1}{\gamma_i}\tilde r_i \dot{\hat r} \notag \\
  \leq &  \frac{\rho_i \varepsilon_i}{2 e_i\left(\rho_i - e_i\right) } \dot \vartheta_i - \frac{\varepsilon _i}{2 \left(\rho_i - e_i\right) } \dot \rho_i  + \frac{|\varepsilon_i| \rho_i}{2e_i(\rho_i-e_i)} \notag \\
  & \times \left(\left\lVert \frac{d h(x)}{d x}\right\rVert + \left\lVert \dot {\hat \upsilon }_i\right\rVert - \left\lVert  \frac{d h(\hat \upsilon_i)}{d \hat \upsilon_i}\right\rVert - \left\lVert \dot {\hat {\bar x}}_i \right\rVert\right)   r_i \notag \\
  & + \frac{|\varepsilon_i| \rho_i \tilde{r}_i}{2e_i(\rho_i-e_i)}\left(\left\lVert \frac{d h(\hat \upsilon_i)}{d \hat \upsilon_i}\right\rVert + \left\lVert \dot {\hat {\bar x}}_i \right\rVert \right) \notag \\
  & +  \frac{\chi_i ^2\hat r_i^2 }{\sqrt{\chi_i ^2\hat r_i^2 + \epsilon_i^2}} + \epsilon_i - \frac{1}{\gamma_i}\tilde r_i \dot{\hat r}_i.
\end{align}
Applying the reverse triangle inequality yields that
\begin{align}
  \label{eq:reverse triangle inequality}
  & \frac{|\varepsilon _i | \rho_i}{2 e_i\left(\rho_i - e_i\right)} \left(\left\lVert \frac{d h(x)}{d x}\right\rVert + \left\lVert \dot {\hat \upsilon }_i\right\rVert - \left\lVert  \frac{d h(\hat \upsilon_i)}{d \hat \upsilon_i}\right\rVert - \left\lVert \dot {\hat {\bar x}}_i \right\rVert\right)   r_i  \notag \\
  & \leq \frac{|\varepsilon _i | \rho_i}{2 e_i\left(\rho_i - e_i\right)} \left(\left\lVert \frac{d h(x)}{d x} - \frac{d h(\hat \upsilon_i)}{d \hat \upsilon_i} \right\rVert + \left\lVert \dot {\hat \upsilon }_i - \dot {\hat {\bar x}}_i\right\rVert \right) r_i  \notag \\
  & \leq \frac{1}{2}\left(\frac{|\varepsilon _i | \rho_i}{2 e_i\left(\rho_i - e_i\right)}\right)^2 + \frac{1}{2}\Delta_i^2,
\end{align}
where $\Delta_i \geq \left(\left\lVert \frac{d h(x)}{d x} - \frac{d h(\hat \upsilon_i)}{d \hat \upsilon_i} \right\rVert + \left\lVert \dot {\hat \upsilon }_i - \dot {\hat {\bar x}}_i\right\rVert \right) r_i$ is a positive constant. Substituting \eqref{eq:rcbf}, (\ref{eq:reverse triangle inequality}) into (\ref{eq:dV_2}) and using Young inequality yields that
\begin{align}
  \label{eq:dV_3}
  \dot V_i \leq & -\frac{c_i}{2} \varepsilon_i^2 +  \frac{\varsigma_i}{\gamma_i} \tilde{r}_i \hat r_i + \frac{1}{2}\Delta_i^2 + \epsilon_i \notag \\
  \leq &  -\frac{c_i}{2} \varepsilon_i^2  -\frac{\varsigma_i}{2\gamma_i} \tilde{r}_i^2 + \frac{\varsigma_i}{2\gamma_i} {r}_i^2 + \frac{1}{2}\Delta_i^2 + \epsilon_i \notag \\
  \leq & - \zeta _i V_i + \Omega_i,
\end{align}
where $\zeta _i = \min\{c_i,\varsigma_i\}$ and $\Omega_i: = \frac{\varsigma_i}{2\gamma_i} {r}_i^2 + \frac{1}{2}\Delta_i^2 + \epsilon_i$ is a positive constant. Integrating both sides of (\ref{eq:dV_3}), it is obtained that 
\begin{equation}
  \label{eq:V_i(t)}
  V_i(t) \leq \left(V_i(0) - \frac{\Omega_i}{\zeta_i}\right) \mathrm{e}^{-\zeta_i t} +  \frac{\Omega_i}{\zeta_i}.
\end{equation}
From (\ref{eq:V_i(t)}), it can be known that $V_i(\bar t^-)$ is bounded, i.e., $\varepsilon_i(\bar t^-)$ is bounded, which is in contradiction with the assumption $\lim_{t \to \bar t}\varepsilon _i(t) = \infty$. Therefore, for $t\in[0,+\infty)$, $\varepsilon_i$ and $\tilde r_i$ remain bounded, i.e., reconstruction error $e_i$ satisfy the prescribed performance criteria.

2) Since for $t\in[0,+\infty)$, we have $e_i > 0$, which implies that $h(x) > \hat h_i(\hat \upsilon_i,\vartheta_i)$ for any $i \in \mathcal{N}$ holds. Thus, if $\hat h_i(\hat \upsilon_i,\vartheta_i) \geq 0$ holds for at least one $i \in \mathcal{N}$, it is deduced that $h(x)>0$ holds.
\end{proof} 

Theorem \ref{the:1}.2) stems from the fact clarified in Remark \ref{re:1}: all $\hat h_i$, $i\in\mathcal{N}$, are reconstructions of the original coupled CBF $h$. Moreover, for scenarios with multiple heterogeneous coupled CBFs, the non-negativity of each original CBF would need to be individually guaranteed by its respective reconstructed CBF remaining non-negative, which can be easily extended from our approach by having each agent reconstruct its own original coupled CBF via (\ref{eq:rcbf}).

\subsection{Design of Distributed Safety Controller}
To guarantee the overall safety of MAS (\ref{eq:sys}), based on the distributed observer (\ref{eq:observer}) and the reconstructed CBF (\ref{eq:rcbf}), a distributed safety controller for agent $i \in \mathcal{N}$ is designed as follows:
\begin{subequations}
  \label{eq:rcbfQP}
\begin{align}
  & u_i = \mathop{\arg \min} \limits_{ {\nu_i}\in \mathbb{R}^m} \,
    \frac{1}{2} (\nu_i - u_{i,\text{nom}}) ^T W_i (\nu_i - u_{i,\text{nom}})   \\ 
  \label{eq:rcbf_constraint}
  &\text{s.t.}  \quad \frac{\partial \hat h_i}{\partial x_i}(f(x_i)+g(x_i)u_i) + \sum_{l \in \mathcal{N}^+, l \neq i} \frac{\partial \hat h_i}{\partial \hat {x}_{i,l}} \dot {\hat {x}}_{i,l} - \dot \vartheta_i  \notag \\
    & \phantom{\text{s.t.}} \quad \geq - \alpha_i(\hat h_i),
\end{align}
\end{subequations}
where $\nu_i$ represents the optimization variable, $u_{i,\text{nom}}$ is the nominal control input of agent $i$, $W_i\in \mathbb{R}^{m \times m}$ is a positive definite weighted matrix, and $\alpha_i(\cdot)$ is an extended class $\mathcal{K}$ function. Then, we have the following theorem.
\begin{theorem}
  \label{the:2}
  Consider the MAS modeled by (\ref{eq:sys}) under Assumption \ref{ass:1}-\ref{ass:5}. With the proposed distributed control scheme (\ref{eq:observer}), (\ref{eq:rcbf}) and (\ref{eq:rcbfQP}), the safe set $\mathcal{C}$ is forward invariant, i.e., the safety of all agents in $\mathcal{N}$ is guaranteed.
\end{theorem}
\begin{proof}
  At initial time $t=0$, according to Assumption \ref{ass:5}, we have $h(x(0))>0$. By appropriately selecting the parameters, we can always ensure that $\hat h_i (\hat \upsilon_i(0), \vartheta _i(0) )\geq 0$ and $ 0 < e_i(0) < \rho_0^i$. As a concrete example, we can set $\rho_0^i = h(x(0))$ and $\vartheta_i(0) = h(\hat \upsilon_i(0))- h(x(0))/2$, then, it follows that $ 0 < e_i(0) =  h(x(0))/2 < h(x(0)) = \rho_0^i$ and $\hat h_i (\hat \upsilon_i(0), \vartheta _i(0) ) = h(x(0))/2 \geq 0$. Furthermore, if $\hat h_i (\hat \upsilon_i(0), \vartheta _i(0) )\geq 0$, from Lemma \ref{lem:1}, it can be known that $\hat h_i (\hat \upsilon_i, \vartheta _i ) \geq 0 $ always holds for $t \in [0, +\infty)$ with the controller (\ref{eq:rcbfQP}). Recalling Theorem \ref{the:1}, it can be inferred that $h(x(t))\geq 0$ also holds for $t \in [0, +\infty)$, which implies that the safety set $\mathcal{C}$ is forward invariant. Therefore, the safety of all agents in $\mathcal{N}$ is rigorously guaranteed.
\end{proof}
\begin{remark}
  The proposed fully distributed scheme rigorously guarantees safety by rendering $\mathcal{C}$ forward invariant (Theorem \ref{the:2}). Crucially, unlike methods \cite{Tan2025,Mestres2024} that require all agents in a coupled constraint to actively comply, our reconstructed CBF approach accommodates uncontrollable agents. As shown in Theorem \ref{the:1}, maintaining the reconstructed constraint (e.g., $\hat h_i \geq 0$) ensures the original coupled constraint $h \geq 0$ is satisfied. Therefore, controllable agents can compensate for uncertain behaviors of uncontrollable agents, requiring only a subset of agents to actively satisfy the constraints.
\end{remark}

\section{SIMULATION}
\label{sec:4}
In this section, the effectiveness of the proposed method is demonstrated by a multi-robot navigation case study. Consider an MAS consisting of four robots with $\mathcal{N}^+ = \{1,2,3,4\}$, where robot $4$ is an uncontrollable agent, while the remaining agents are controllable. The dynamics of robots are given as
\begin{align}
  \label{eq:rb}
  \begin{cases}
    \dot {\boldsymbol{p}}_i = \left[ \begin{matrix}
    \cos \theta_i & -l \sin \theta_i \\
    \sin \theta_i & l \cos \theta_i
\end{matrix}\right] u_i , \\ 
    \dot \theta_i = \omega_i , 
\end{cases} i = 1,2,3,4,
\end{align}
where ${\boldsymbol{p}}_i = [x_i,y_i]^T$ denotes the position of the center of robot $i$ and $\theta_i$ is its heading angle. $u_i = [v_i, \omega_i]^T$ is the control input of robot $i$, where $v_i$ and $\omega_i$ are the linear velocity and angular velocity of robot $i$, respectively. $l=0.036$ is the distance between the center point ${\boldsymbol{p}}_i$ and the wheel axle of robot $i$. As the subsequent robotic task depends only on ${\boldsymbol{p}}_i$, (\ref{eq:rb}) can be regarded as (\ref{eq:sys}) with $f= \boldsymbol{0}$ and $g=[\cos \theta_i, -l\sin \theta_i; \sin \theta_i, l\cos \theta_i]$. The communication topology is illustrated in Fig. \ref{fig:1}. 

\begin{figure}[!htb]
\centering
\includegraphics[width=4cm]{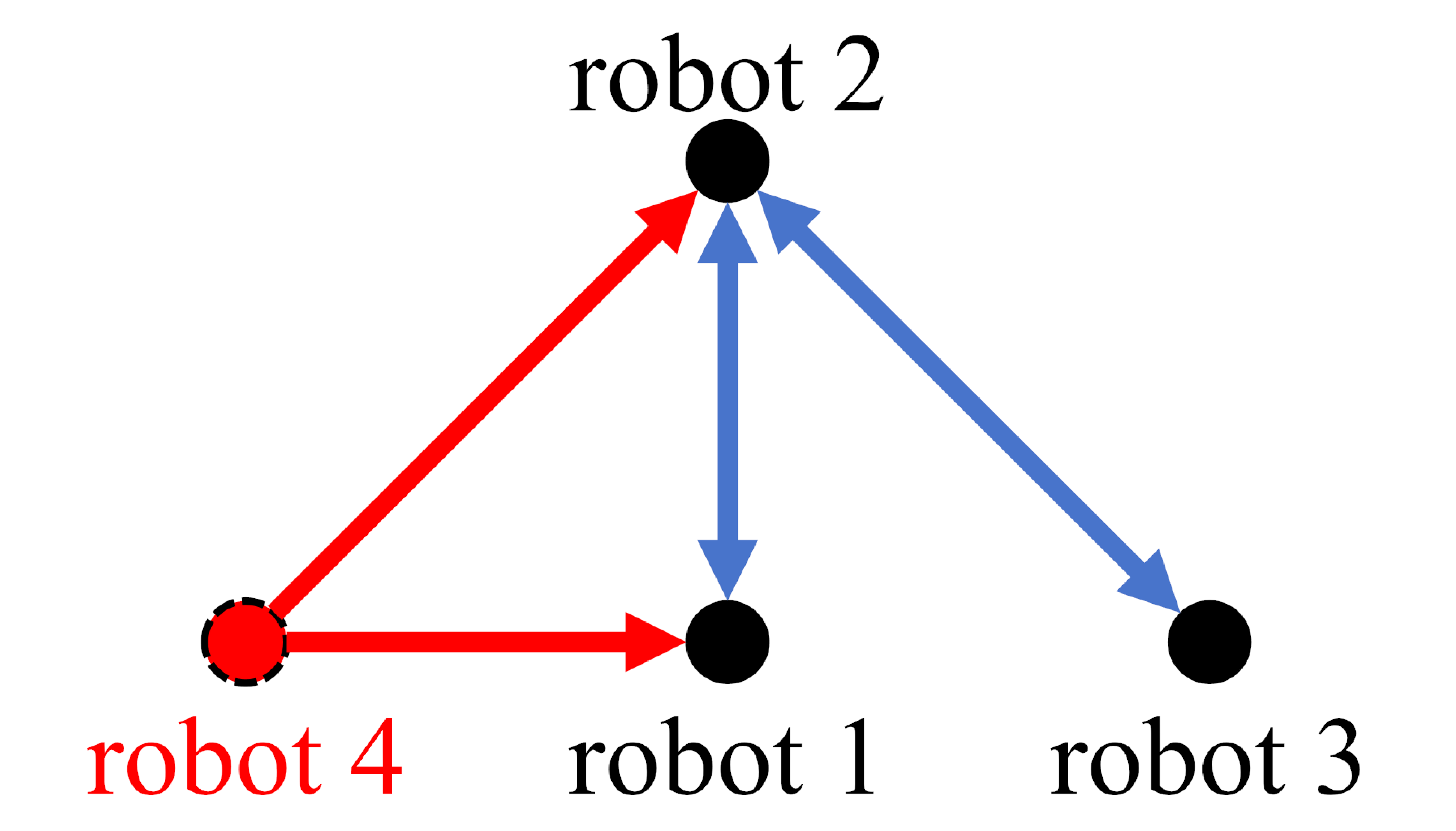}
\caption{Communication topology.}
\label{fig:1}
\end{figure}

The robots' working environment contains three circular obstacles $O_i$, $i=1,\dots,3$. The center position ${\boldsymbol{p}}_{O,i}$ and the radius $r_{O,i}$ of $O_i$ are set as: $[{\boldsymbol{p}}_{O,1}^T,r_{O,1}] = [0.8,2.5,0.5]$, $[{\boldsymbol{p}}_{O,2}^T,r_{O,2}] = [3,3.5,0.3]$, $[{\boldsymbol{p}}_{O,3}^T,r_{O,3}] = [3,1.5,0.5]$. Besides, the environment includes a goal point $G_0$ located at $\boldsymbol{p}_{G,0}=[0,4]^T$ and a circular target region $G_1$ centered at $\boldsymbol{p}_{G,1}=[1,0.5]^T$ with a radius $r_{G,1} = 0.4$. Robot $4$ is required to navigate to the goal point $G_0$ and robots $1,2$ are tasked with navigating to the goal region $G_1$. In addition, robots $i$ for $i=1,2,3$ are subject to the following constraints:
\begin{itemize}
  \item \textit{Collision avoidance}: Robots $i \in \{1,2,3\}$ must avoid obstacles and other robots, yielding the constraints: 
  \begin{align*}
    & b_{i,O_l} = ||\boldsymbol{p}_{i} - \boldsymbol{p}_{O,l}||^2 - (r_{R}+r_{O,l})^2 \geq 0, \,l\in\{1,2,3\}, \\
    & b_{i,j} = ||\boldsymbol{p}_{i} - \boldsymbol{p}_{j}||^2 - 4r_{R}^2 \geq 0,\, j\in\{1,2,3,4\}, i\neq j,
  \end{align*}
  where $r_R = 0.1$ is the collision radius of the robots. 
  \item \textit{Coupled constraints}: Robots $1$ and $2$ must cooperate, while robot $3$ follows the uncontrollable robot $4$. Both pairs must maintain a distance within $d_f=1.25$, yielding:
  \begin{align*}
    & b_{i} = d_f^2 - ||\boldsymbol{p}_{i} - \boldsymbol{p}_{j}||^2 \geq 0, \, i\neq j, \, i,j\in\{1,2\}, \\
    & b_{3} = d_f^2 - ||\boldsymbol{p}_{3} - \boldsymbol{p}_{4}||^2 \geq 0.
  \end{align*}
\end{itemize}
Then, we derive the following CBF:
\begin{equation*}
  h_i = -\frac{1}{20}\ln\left( \sum_{l=1}^{3}\mathrm{e}^{-20b_{i,O_l}} + \sum_{j\neq i, j \in \mathcal{N}^+}\mathrm{e}^{-20b_{i,j}} + \mathrm{e}^{-20b_{i}}\right).
\end{equation*}
It is deduced that $h_i \leq \min_{l\in \{1,2,3\}, j\neq i, j\in \mathcal{N}^+}\{b_{i,O_{l}},b_{i,j},b_{i}\}$. 
In the considered case study, there exists multiple heterogeneous coupled CBFs, i.e., $h_i$, $i=1,2,3$. Nevertheless, we can still reconstruct $h_i$ using (\ref{eq:rcbf}) and subsequently establish a safety controller based on (\ref{eq:rcbfQP}).

\begin{figure}[!htb]
\centering
\includegraphics[width=8cm]{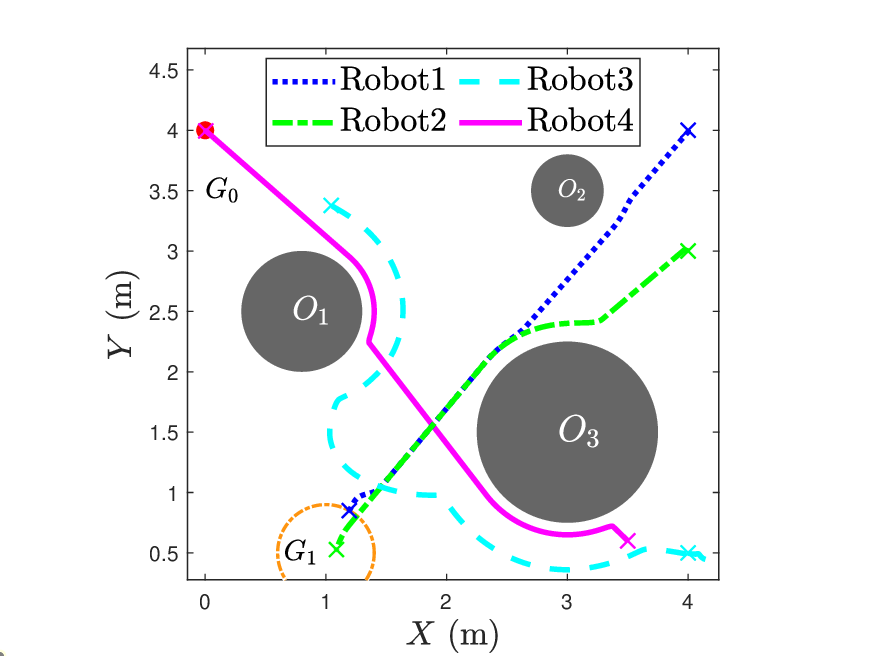}
\caption{Trajectories of 4 robots.}
\label{fig:2}
\end{figure}

\begin{figure}[!htb]
\centering
\includegraphics[width=8.87cm]{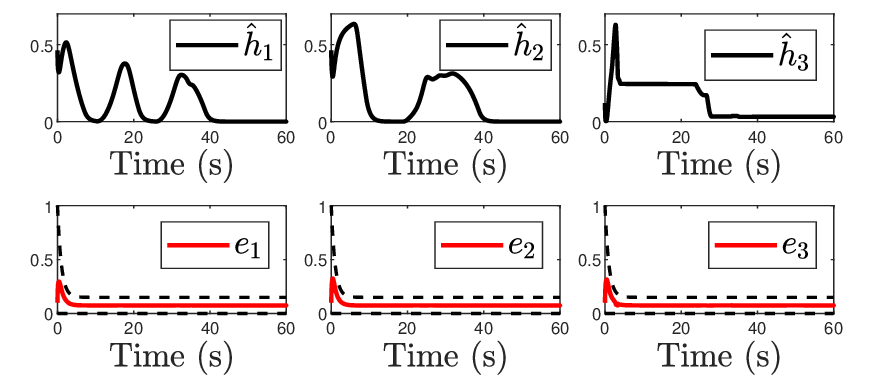}
\caption{Reconstructed CBF $\hat h_i$ and reconstruction error $e_i$, $i=1,\dots,3$.}
\label{fig:3}
\end{figure}

The initial states of the robots are set as follows: $[\boldsymbol{p}_1(0)^T,\theta_1(0)] = [4,4,0]$, $[\boldsymbol{p}_2(0)^T,\theta_2(0)] = [4,3,0]$, $[\boldsymbol{p}_3(0)^T,\theta_3(0)] = [4,0.5,-\pi]$, $[\boldsymbol{p}_4(0)^T,\theta_4(0)] = [3.5,0.6,-\pi]$. The distributed observers here are used solely for estimating $\boldsymbol{p}_i$ and the distributed observers parameters are set as: the initial estimate $\hat x_{i,j}$ is set to be identical to the actual initial state, $\hat \delta_{i,j}(0) = 2$, $\sigma_j = 0.01$, $P_j = \mathrm{diag}\{2,2\}$ for $i = 1,2,3, j \in \mathcal{N}^+$. The parameters for the reconstructed CBFs are set as follows: $\vartheta_i(0) = 0.1$, $\hat r_i(0) = 0$, $\rho_0^i = 1$, $\rho_{\infty}^i = 0.15$, $\varrho = 1$, $c_i = 0.01$, $\varsigma_i = 0.8$, $\gamma_i = 0.01$, $\epsilon_i=0.01$ for $i=1,2,3$. For QP controller, we choose $\alpha_1(\hat h_1) = \hat h_1$, $\alpha_2(\hat h_2) = \hat h_2$, $\alpha_3(\hat h_3) =\hat h_3^5 /10$, and $W_i = \mathrm{diag}\{1,l^2\}$ for $i=1,2,3$. Additionally, the constraints on the robots' maximum linear velocity and angular velocity are imposed, such that $|v_i|\leq v_{\max}$ and $|\omega_i|\leq \omega_{\max}$ for $i \in \mathcal{N}^+$, where $[v_{\max},\omega_{\max}] = [0.22,2.84]$.

\begin{figure}[!htb]
\centering
\includegraphics[width=8.87cm]{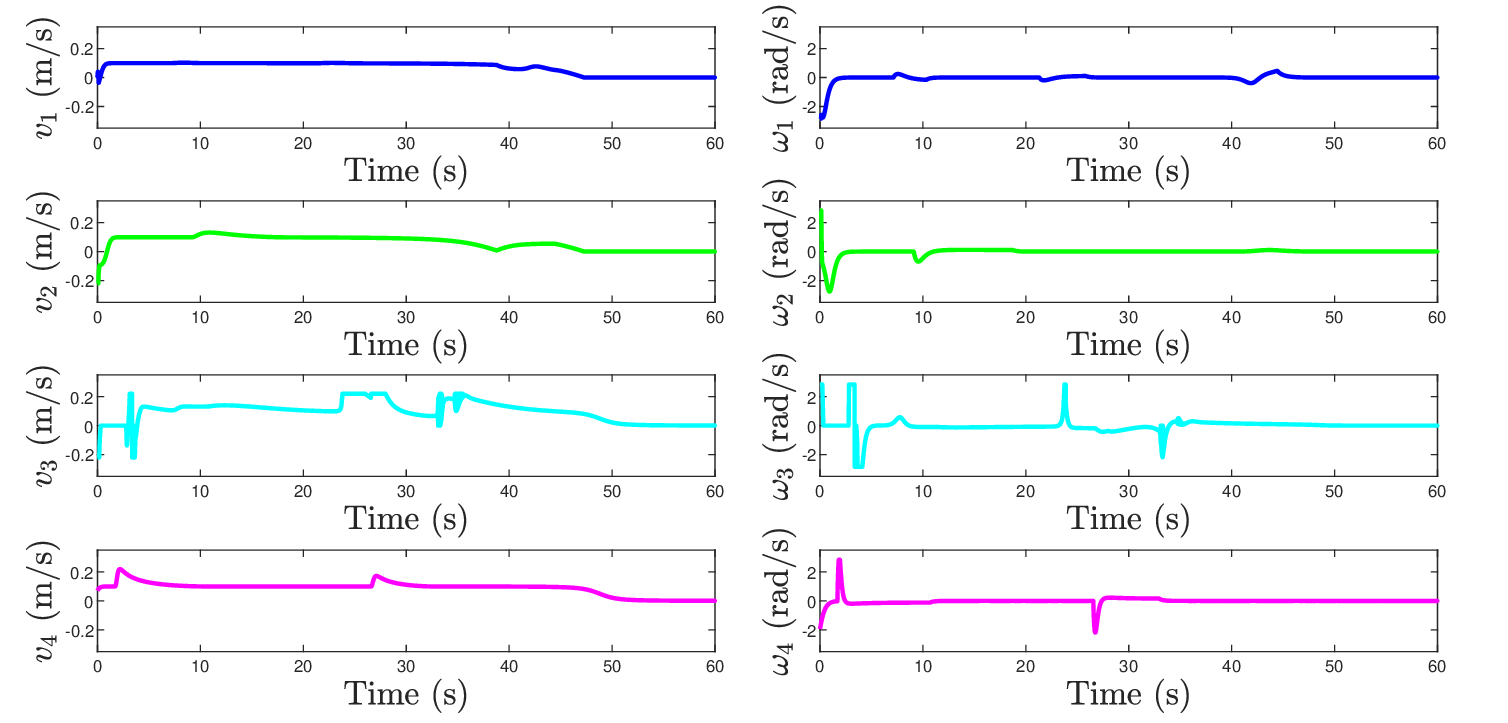}
\caption{Control input $u_i$, $i=1,\dots,4$.}
\label{fig:4}
\end{figure}

The trajectories of the robots are presented in Fig. \ref{fig:2}. Fig. \ref{fig:3} illustrates the reconstructed CBFs $\hat h_i$ and reconstruction errors $e_i$, $i = 1,2,3$, where the black dashed line represents the prescribed performance boundary. The control inputs for each robot are displayed in Fig. \ref{fig:4}. As shown in Fig. \ref{fig:3}, all reconstructed CBFs remain non-negative and the reconstruction error satisfies the prescribed performance requirements. This implies that the original constraint $h_i \geq 0$ is also satisfied. Therefore, the proposed distributed control method effectively ensures system safety.

\section{CONCLUSION}
\label{sec:5}
This paper presented a distributed safety critical control scheme based on reconstructed CBFs to address the control problem of MASs subject to coupled safety constraints. The original coupled CBF is reconstructed using state estimates obtained from a distributed adaptive observer, which enables the resulting reconstructed CBFs to be fully distributed.
The reconstruction process is modified by designing a prescribed performance adaptive parameter, such that the satisfaction of the reconstructed CBF constraints guarantees that of the original coupled one.
Based on the reconstructed CBFs, a safety QP controller is designed and we prove that this controller strictly guarantees the safety of the MAS. Compared to previous works, our proposed method is also applicable to uncertain dynamic environments containing uncontrollable agents. A simulation is conducted to demonstrate the effectiveness of the proposed method.

\addtolength{\textheight}{-12cm}   % This command serves to balance the column lengths

%%%%%%%%%%%%%%%%%%%%%%%%%%%%%%%%%%%%%%%%%%%%%%%%%%%%%%%%%%%%%%%%%%%%%%%%%%%%%%%%

%%%%%%%%%%%%%%%%%%%%%%%%%%%%%%%%%%%%%%%%%%%%%%%%%%%%%%%%%%%%%%%%%%%%%%%%%%%%%%%%

% \section*{proof of Proposition \ref{pro:1}}

% \section*{ACKNOWLEDGMENT}

% The preferred spelling of the word  acknowledgment  in America is without an  e  after the  g . Avoid the stilted expression,  One of us (R. B. G.) thanks . . .   Instead, try  R. B. G. thanks . Put sponsor acknowledgments in the unnumbered footnote on the first page.

% %%%%%%%%%%%%%%%%%%%%%%%%%%%%%%%%%%%%%%%%%%%%%%%%%%%%%%%%%%%%%%%%%%%%%%%%%%%%%%%%

% References are important to the reader; therefore, each citation must be complete and correct. If at all possible, references should be commonly available publications.

\end{document}